\documentclass[11pt]{article}

\usepackage{cmap}                    
\usepackage{mathtext}                 
\usepackage[T2A]{fontenc}            
\usepackage[utf8]{inputenc}            
\usepackage[english]{babel}    
\usepackage{upgreek}
\usepackage{amsmath,amsfonts,amssymb,amsthm,mathtools} 
\usepackage{icomma} 
\title{Note on generalized adaptive group testing}
\usepackage{authblk}

\usepackage{graphicx}  
\usepackage{wrapfig} 

\usepackage{array,tabularx,tabulary,booktabs} 
\usepackage{longtable}  
\usepackage{multirow} 

\usepackage{hyperref}
\usepackage{url}
\usepackage{color}

\newcommand{\nothing}[1]{}

\newcommand{\beq}[1]{\begin{equation}\label{#1}}
	\newcommand{\eeq}{\end{equation}}

\newcommand{\bmu}[1]{\begin{multline}\label{#1}}
	\newcommand{\emu}{\end{multline}}

\newcommand{\x}{{\textbf{\textit{x}}}}

\newcommand{\y}{{\textbf{\textit{y}}}}

\renewcommand{\geq}{\geqslant}
\renewcommand{\leq}{\leqslant}
\theoremstyle{plain} 
\newtheorem{theorem}{Theorem}
\newtheorem{proposition}{Proposition}
\newtheorem*{theorem*}{Theorem}
\author{Ilya Vorobyev}
\affil{Technical University of Munich, Munich, Germany}
\date{}       
\title{Note on generalized group testing}
\begin{document}
	\maketitle

\begin{abstract}
    In this note, we present a new adaptive algorithm for generalized group testing, which is asymptotically optimal if $d=o(\log_2|E|)$, $E$ is a set of potentially contaminated sets, $d$ is a maximal size of elements of $E$. Also, we design a 3-stage algorithm, which is asymptotically optimal for $d=2$. 
\end{abstract}

\section*{Introduction}
Group testing~\cite{dorfman1943detection} is a combinatorial problem where one needs to identify the set of defective $d$ elements among the population of $M$ elements. To achieve this goal it is allowed to test arbitrary subsets. The test result is positive if the subset contains at least one defective element; otherwise, it's negative. The goal is to find all defectives by using a minimal number of tests.

Group testing problem can be described with the language of graph theory. Let's say that we have a hypergraph $H=(V, E)$, $|V|=M$, and the set of edges $E$ consists of all possible subsets of cardinality $d$. Our goal is to identify one defective edge $e\in E$ with the help of special tests. Each test is a subset of vertices, the test result is positive if the tested set intersects defective edge $e$; otherwise, the test result is negative. It is easy to see that this search problem of an edge is equivalent to the classical group testing problem.
The natural generalization of this problem is to consider an arbitrary hypergraph $H$. Such a problem was called group testing on general set-systems in paper~\cite{gonen2022group}. The authors proved that for the adaptive setting(i.e. each test can depend on the results of previous tests) and $d$-uniform hypergraph $H$ the defective edge can be found with $O(\log |E| + d\log^2d)$ tests. 

We provide an algorithm  which uses only $\log_2|E|+2\sqrt{d\log_2|E|}+O(d)$ tests. It means that for $d=o(\log_2|E|)$ the number of tests is $\log_2|E|(1+o(1))$. For $d=\Omega(\log_2E)$ the number of tests $O(d)$. Recall the lower bound $\log_2|E|+\Omega(d)$ from~\cite{gonen2022group}. In the first case, our algorithm is optimal up to $1+o(1)$ factor, and in the second -- up to a constant factor.

In addition, we show that for $d=2$ it is possible to find the defective edge with $\log_2|E|(1+o(1))$ tests by using a 3-stage algorithm.

\section*{Adaptive algorithm}
\begin{theorem}
Let $H$ be a $d$-uniform hypergraph with one defective edge $e$. It is possible to find this defective edge by using $\log_2|E|+2\sqrt{d\log_2|E|}+O(d)$ adaptive tests.
\end{theorem}
\begin{proof}
Sort all vertices of the hypergraph $H$ by their degrees $d_1\leq d_2\leq\ldots\leq d_M=N_1\leq |E|$.
Divide all vertices into $N_2=\lfloor\log_2{N_1}\rfloor+1$ groups. In group $i$ we include all vertices with degree $d_v$ such that $2^{i-1}\leq d_v<2^i$.
Denote this groups as $V_1$, $V_2$, \ldots, $V_{N_2}$. Now we divide this groups into $N_3=\lceil N_2/f\rceil$ sets $W_1$, $W_2$, \ldots, $W_{N_3}$, where $f$ equals $\max\left(1, \left\lceil\sqrt{\frac{\log_2|E|}{d}}\right\rceil\right)$.
The first set $W_1$ is a union of the first $f$ groups $V_1$, \ldots, $V_{f}$, the second set $W_2$ is a union of the second $f$ groups, and so on. The last set $W_{n_3}$ can have a smaller size.

We start our algorithm by testing sets $W_i$ in reverse order. After we obtain a positive result for some $W_i$, we start testing sets $V_j$, which belong to $W_i$. Again, we test them in reverse order, i.e. groups containing vertices with bigger degrees are tested earlier. At some point, we will obtain a positive result. Thus, we have a set $V_i$, which contains at least one defective element. Find this defective element with $\lceil\log|V_i|\rceil$ adaptive tests.
Erase all edges from the hypergraph $H$, which are not compatible with at least one obtained test result. Redistribute vertices of the hypergraph into sets $V_i$ and $W_j$ according to their new degrees. Note that a hypergraph vertex can only change its set $V_i$ or $W_j$ to a set with a smaller index.
Repeat the procedure with modification, that there is no need to test set $W_i$ or $V_j$ if we have already obtained a negative outcome for it.

Calculate the total number of tests. The number of negative tests on $W_i$ and $V_j$ is less than $N_3+df$. The number of positive tests on $W_i$ and $V_j$ is $2d$. The number of tests used to adaptively find elements in $V_i$, $|V_i|=n$, is close to optimal since all have the same degree. More formally, let's say that all vertices from $V_i$ have degree $d_j$, $k\leq d_j<2k$. It means that before we start testing the following inequality $kn\leq|E|<2kn$ holds. After the procedure, the number of remaining edges $E'$ satisfies $|E'|<2k$. So, we spent $<\log_2n+1$ tests and reduced $\log |E|$ by at least $\log_2n - 1$, which is optimal up to a constant addend. These constants give $O(d)$ additional tests at the end. Therefore, the total number of tests is
$$
\log_2|E|+2\sqrt{d\log_2|E|}+O(d)
$$
for $d<\log_2|E|$ and $O(d)$ for $d\geq \log_2|E|$.
\end{proof}

\section*{Non-adaptive algorithm}

For a non-adaptive group testing problem, upper bounds on the number of tests are proved with the help of the probabilistic method. All these bounds\cite{d2014lectures,erdHos1985families,quang1988bounds,dyachkov1989superimposed} have the same asymptotic $O(d^2\log_2M)$, but the hidden constants are different. All these bounds can be trivially generalized for the case of group testing on general system sets. In paper~\cite{gonen2022group} the authors proved an upper bound $O(d\log_2|M|)$ by using a Bernoulli ensemble. The random coding with constant weight codes \cite{dyachkov1989superimposed} gives a better constant. Below we state a result for generalized group testing analogous to the result for traditional group testing problem from~\cite{dyachkov1989superimposed}.

\begin{theorem}
Let $H=(V, E)$ be a hypergraph with a maximal size of an edge $d$. Then the number of tests needed to non-adaptively find a defective edge in $H$ is at most $d\log_2e\log_2|E|(1+o(1))$.  
\end{theorem}

We omit the proof since it is a trivial generalization of the result from~\cite{dyachkov1989superimposed}.

\section*{Optimal 3-stage algorithm for $d=2$}
For the special case of $d=2$, it is possible to find the defective edge with an optimal number of tests by using a 3-stage algorithm.
\begin{theorem}
    Let $G=(V, E)$ be an arbitrary graph with one defective edge. It is possible to find this defective edge with $\log_2|E|(1+o(1))$ tests by using a 3-stage algorithm.
\end{theorem}
\begin{proof}

Let's describe our algorithm. It starts in the same way as the adaptive algorithm.

Sort all vertices of the graph $G$ by their degrees, and denote the maximal degree as $N_1$. Divide all vertices into $N_2=\lfloor\log_2{N_1}\rfloor+1$ groups. In group $i$ we include all vertices with degree $d_v$ such that $2^{i-1}\leq d_v<2^i$.
Denote this groups as $V_1$, $V_2$, \ldots, $V_{N_2}$.

Use some non-adaptive algorithm to identify sets $V_i$ which contain defective elements. It requires at most $O(\log N_2)=O(\log\log |E|)$ tests.
We may obtain one or two positive results. The case with two positive results is trivial. Say that sets $V_{i_1}$ and $V_{i_2}$ contain defective elements, $|V_{i_j}|=n_j$, with a slight abuse of notation we say that degrees of vertices from $V_{i_j}$ is in $[d_j, 2d_j)$. Obviously, $|E|\geq n_1d_1$. In the second stage, we non-adaptively find one defective element in the set $V_{i_1}$ by using $\lceil \log_2 n_1 \rceil$ tests. In the third stage, we find the second defective among at most $2d_1$ neighbors of the first one by using at most $\lceil \log_22d_1\rceil$. The total number of tests used in the first, second, and third stages is at most $$
O(\log_2\log_2 |E|)+\lceil \log_2 n_1 \rceil + \lceil \log_22d_1\rceil=\log_2|E|(1+o(1)).
$$

Now we proceed to a more complicated case when only one set $V_i$ contains defective. We use the same idea as in the paper~\cite{vorobyev2019new}, where a 2-stage algorithm to find 2 defectives in a traditional setting was proposed. It turns out that we can use almost the same proof because it is only important that all vertices have approximately the same degrees. We provide modified proof for completeness.

Define $E'$ as all edges of $E$, both endpoints of which belong to $V_i$. Denote the cardinality of $V_i$ as $n$, degrees of all vertices from $V_i$ in the graph $G$ is in $[d, 2d)$.
Consider a random matrix $X$ of size $T\times n$, each column $\x_i$ of which is chosen independently and uniformly from the set of all columns of weight $wT$. We ignore the fact that this is not necessarily an integer, it will not affect our result. For any vector $\y \in \{0, \; 1\}^T$ define a graph $G(H, X, \y)=(V, E_y)$, which contains all edges $e=(v_1, v_2)$ from $E'$, such that the union of columns $\x_{v_1}$ and $\x_{v_2}$ equals $\y$. 

Let $L$ be some slowly growing function of $|E|$, for example $L=\log_2\log_2|E|$. Say that an index $v\in [n]$ is $\y$-bad index of the first type if the degree of the vertex $v$ in the graph $G_y$ is at least $L$. Call an index $v\in [n]$ a $\y$-bad index of the second type if in the graph $G_y$ the vertex $v$ is included in some matching of size at least $L$. At last, call an index $v\in [n]$ bad if it is a bad index of the first or second type.

Let's estimate the mathematical expectation of the number of bad indices. Denote the event that a fixed index $v$ is a bad index of the first (second) type for some vector $\y$ as $B_{v, \y, 1}$ ($B_{v, \y, 2}$).  We upper bound the probability $Pr(B_{v, \y, 1})$ by the probability that there exists a non-ordered collection of $L$ other vertices, such that the graph $G_y$ contains edges $(v, v_i)$ for $i=1, \ldots, L$. Hence,
\begin{equation}
    Pr(B_{v, \y, 1})\leq \binom{d_v}{L}p_1(\y)^L<(d_vp_1(\y))^L<(2\sqrt{|E|}p_1(\y))^L,
\end{equation}
where $d_v$ is a degree of the vertex $v$ in the graph $G_y$, $p_1=\binom{wT}{(q-w)T}\bigg/ \binom{T}{wT}$ is a probability that the union of $\x_{v}$, and a random column of weight $w$ equals to the vector $\y$ of weight $qT$, which covers vector $\x_{v}$. The last inequality holds since
$$
d_v^2<d_vn\leq 2dn\leq 4|E|.$$

The probability $Pr(B_{v, \y, 2})$ can be upper bounded by the probability that there exists a vertex $v_1$ such that an edge $(v, v_1)\in G_y$ and $L-1$  edges $(v_{2i}, v_{2i+1})\in G_y$ for $i=1,\ldots, L-1$, such that all these edges don't intersect each other.
\begin{equation}
    Pr(B_{v, y, 2})\leq 2d \binom{|E|}{L-1}p_1(\y)p_2(\y)^{L-1}<|E|^{L}p_2(\y)^{L-1},
\end{equation}
where $p_2(\y)=\frac{\binom{qT}{wT}\binom{wT}{(q-w)T}}{\left(\binom{T}{wT}\right)^2}$ is a probability that the union of two random vectors of weight $wT$ equals to the vector $\y$ of weight $qT$.

The mathematical expectation of the number of bad indices is at most
\begin{align}
n2^T\sup\limits_q ((2\sqrt{|E|}p_1(\y))^L+(|E|^{L}p_2(\y)^{L-1}))
\end{align}

Take $w=1-\sqrt{2}/2$, $T=\log_2|E|\frac{L+3}{L-1}$. It is easy to check that $p_1(\y)\leq 2^{(-0.57 +o(1))T}$, $p_2(\y)\leq 2^{(-1+o(1))T}$.
Using the following 3 obvious inequalities
\begin{equation}
    n2^T<|E|^{2+o(1)}
\end{equation}
\begin{equation}
    (2\sqrt{|E|}p_1(\y))^L\leq |E|^{(-0.07+o(1))L}
\end{equation}
\begin{equation}
    |E|^Lp_2(\y)^{L-1}\leq |E|^{L-\frac{L+3}{L-1}(L-1)}=|E|^{-3}.
\end{equation}
we conclude that the mathematical expectation of the number of bad indices is at most $|E|^{-1+o(1)}$. It means that there exists a matrix $X$ without bad indices. Use such a matrix as a testing matrix in the second stage.

Then we use the following simple proposition, which proof can be found in, for example, \cite{vorobyev2019new}.
\begin{proposition}
If the maximum vertex degree and the maximum cardinality of a matching in a graph $G = (V, E)$ are
less than $L$, then $|E| < 2L^2$.
\end{proposition}

It means that after the second stage we have at most $2L^2$ edges. Therefore, we can test all non-isolated vertices with $\leq 4L^2=4\log_2\log_2|E|$ tests.

The total number of tests is $O(\log\log|E|)+\log_2|E|\frac{L+3}{L-1}+4L^2=\log_2|E|(1+o(1))$.
\end{proof}

	\bibliographystyle{plain}
	\bibliography{refs}{}

\begin{thebibliography}{1}

\bibitem{dorfman1943detection}
Robert Dorfman.
\newblock The detection of defective members of large populations.
\newblock {\em The Annals of mathematical statistics}, 14(4):436--440, 1943.

\bibitem{dyachkov1989superimposed}
A~Dyachkov, V~Rykov, and A~Rashad.
\newblock Superimposed distance codes.
\newblock {\em Problems of Control and Information Theory}, 18(4):237--250,
  1989.

\bibitem{d2014lectures}
Arkadii~G D'yachkov.
\newblock Lectures on designing screening experiments.
\newblock {\em arXiv preprint arXiv:1401.7505}, 2014.

\bibitem{erdHos1985families}
Paul Erd{\H{o}}s, Peter Frankl, and Zolt{\'a}n F{\"u}redi.
\newblock Families of finite sets in which no set is covered by the union of r
  others.
\newblock {\em Israel J. Math}, 51(1-2):79--89, 1985.

\bibitem{gonen2022group}
Mira Gonen, Michael Langberg, and Alex Sprintson.
\newblock Group testing on general set-systems.
\newblock {\em arXiv preprint arXiv:2202.04988}, 2022.

\bibitem{quang1988bounds}
T.~Zeisel Nguyen Quang~A.
\newblock Bounds on constant weight binary superimposed codes.
\newblock {\em Problems of Control and Information Theory}, 17(4):223--230,
  1988.

\bibitem{vorobyev2019new}
Ilya Vorobyev.
\newblock A new algorithm for two-stage group testing.
\newblock In {\em 2019 IEEE International Symposium on Information Theory
  (ISIT)}, pages 101--105. IEEE, 2019.

\end{thebibliography}
	
\end{document}